\documentclass[preprint,12pt]{elsarticle}

\usepackage{amsmath,amssymb,amsthm}
\usepackage{mathtools}
\usepackage{booktabs}
\usepackage{xcolor}
\usepackage{tikz}
\usepackage{algorithm}
\usepackage{algpseudocode}
\usepackage[colorlinks=true,linkcolor=blue!70!black,
            citecolor=blue!70!black,urlcolor=blue!70!black]{hyperref}
\usepackage[capitalise]{cleveref}

\newtheorem{theorem}{Theorem}[section]
\newtheorem{lemma}[theorem]{Lemma}
\newtheorem{proposition}[theorem]{Proposition}
\newtheorem{corollary}[theorem]{Corollary}
\newtheorem{definition}[theorem]{Definition}
\newtheorem{remark}[theorem]{Remark}

\newcommand{\T}{\mathcal{T}}
\newcommand{\RC}{\mathsf{R}}   

\newcommand{\Rr}{\mathbb{R}}

\newcommand{\stack}{\Pi}
\newcommand{\GL}{\mathcal{G}_L}
\newcommand{\relay}[2]{\hat{\gamma}_{#1,#2}}
\DeclareMathOperator{\MI}{I}
\DeclareMathOperator{\En}{H}

\journal{Neural Networks}

\begin{document}

\begin{frontmatter}

\title{The Extremum Stack as Optimal Memory for\\
  Rate-Independent Sequence Models:\\
  Information-Theoretic Foundations and Online Complexity}

\author[pwut]{Piotr Frydrych}
\ead{piotr.frydrych@pw.edu.pl}

\affiliation[pwut]{organization={Faculty of Mechatronics,
  Warsaw University of Technology},
  city={Warsaw}, country={Poland}}

\begin{abstract}
A functional on discrete sequences is \emph{rate-independent} if it
is invariant under monotone time reparametrisations: it responds only
to the order of local extrema of its input, not to their timing.  The
central algorithmic object of this class is the \emph{Preisach
extremum stack} $\stack_n$: the nested sequence of alternating local
maxima and minima of $u_{0:n} \in \GL^{n+1}$ that survive the
classical wiping-out rule of the Preisach hysteresis operator.  We
give a complete information-theoretic and algorithmic
characterisation of this object.

First, a \emph{characterisation theorem}: a computable functional $F$
is rate-independent if and only if $F[u](n) = f(\stack_n)$ for a
computable $f$; the stack is a complete invariant of
rate-independence.

Second, \emph{Kolmogorov minimality}:
$K(\stack_n) - O(1) \leq K_{\RC}(u_{0:n}) \leq K(\stack_n) + O(1)$,
where $K_{\RC}(u_{0:n})$ is the length of the shortest program
answering every query in the class $\RC$ of computable rate-independent
functionals, and the $O(1)$ overhead is independent of both the
sequence length $n$ and the stack depth $k$ (it depends only on the
grid resolution $L$).  Any compression of a
hysteresis-driven stream that preserves the full class $\RC$ must
retain at least $K(\stack_n) - O(1)$ bits.

Third, \emph{Shannon minimality}: under any probability measure on
$u_{0:n}$, $\MI(u_{0:n};\stack_n) \leq \MI(u_{0:n};S)$ for every
random variable $S$ from which all $\RC$-queries are computable, with
equality characterising $S$ as informationally equivalent to
$\stack_n$.

Fourth, \emph{worst-case update complexity}: (i)~any compact exact
$\RC$-minimal representation incurs $\Theta(k)$ output changes per
step in the worst case, in a model-independent output-change metric;
(ii)~the monotone ordering of the Preisach wiping property enables
binary search, reducing \emph{boundary detection} to $O(\log k)$,
though physical deletion remains $\Theta(d)$; (iii)~a finger-tree
implementation achieves $O(\log k)$ worst-case time per step for both
search \emph{and} deletion, while maintaining exact $\RC$-minimality
with no approximation error.

Together these results identify the extremum stack as the canonical
minimal representation for rate-independent computation --- in the
worst-case (Kolmogorov), average-case (Shannon), and online
algorithmic senses --- and settle its maintenance complexity.
\end{abstract}


\end{frontmatter}

\section{Introduction}
\label{sec:intro}

The Preisach Attention Layer (PAL) \citep{Frydrych2026PAL} is a
recently proposed sequence modelling architecture that replaces
softmax attention with a stack of binary relays from classical
hysteresis theory.  PAL is Turing-complete at depth $O(1)$, and its
function class is provably incomparable with the transformer's: the
separating property is \emph{rate-independence} --- PAL responds only
to the ordered sequence of local extrema of its input, not to
absolute token positions or temporal spacing.  By discarding timing
and retaining only the alternation structure, PAL reduces per-token
inference from $O(n^2)$ to $O(n \log n)$ and replaces the
$O(n \cdot d_{\mathrm{model}})$ KV-cache with an
$O(k \cdot d_{\mathrm{model}})$ extremum stack, where $k \leq n$ is
the current stack depth.

The central algorithmic object of PAL is this \emph{extremum stack}
$\stack_n$: the nested sequence of alternating local maxima and minima
of the input $u_{0:n} \in \GL^{n+1}$ that survive the classical
\emph{wiping-out rule} \citep{Mayergoyz1991} --- a new extremum
erases all prior extrema of smaller magnitude.  The stack is
maintained online by a simple algorithm (\cref{alg:stack}) running in
$O(1)$ amortised time per step.  Rate-independence --- and the
extremum stack --- also appears in the modelling of ferromagnetic
materials \citep{Brokate1996,Mayergoyz1991}, elastoplastic systems
\citep{Visintin1994}, and financial market threshold models
\citep{Frydrych2014}; PAL is the first architecture to import this
structure into sequence modelling.

The companion paper \citep{Frydrych2026PAL} introduces PAL and proves
the architecture-level results (Turing completeness, expressiveness
separation, logical characterisation, RFIM equivalence).  The present
paper establishes the \emph{information-theoretic and algorithmic
foundations} of the extremum stack itself.  We ask, and answer, three
questions.

\begin{enumerate}
\item[\textbf{Q1.}] \emph{Is the stack the right internal state for
  PAL?}  We prove a characterisation theorem (\cref{thm:charact}): a
  computable functional is rate-independent \emph{if and only if} it
  factors through the stack.  The stack is thus a complete invariant
  of rate-independence, not merely a convenient encoding of Preisach
  outputs.  For PAL, this means that the extremum stack is the unique
  minimal internal state from which every rate-independent query can
  be answered --- any alternative state that preserves PAL's function
  class must contain at least as much information.
\item[\textbf{Q2.}] \emph{Is the stack the smallest summary?}
  We prove minimality in two orthogonal senses.  In the
  \emph{Kolmogorov} sense (\cref{sec:kolmogorov}), which holds
  worst-case for every individual sequence, the shortest program
  answering all rate-independent queries has length
  $K(\stack_n) \pm O(1)$, with the additive constant independent of
  both the sequence length $n$ and the stack depth $k$.  In the
  \emph{Shannon} sense (\cref{sec:shannon}), which holds in
  expectation under any probability measure, the stack has the least
  mutual information with the input among all sufficient
  representations, and is essentially the unique minimiser.  Neither
  result implies the other; together they identify $\stack_n$ as the
  canonical information-minimal representation for the class $\RC$.
\item[\textbf{Q3.}] \emph{What is the per-token cost of maintaining
  the stack?}  The standard update algorithm is $O(1)$ amortised, but
  adversarial inputs force $\Theta(k)$ operations in a single step
  --- a latency spike that is unacceptable for real-time inference.
  We establish a sharp three-level worst-case picture
  (\cref{sec:complexity}): a model-independent $\Theta(k)$ lower
  bound on \emph{output changes} for any compact exact
  representation; an $O(\log k + d)$ binary-search variant showing
  that the wiping property reduces search but not deletion cost; and
  a finger-tree implementation achieving $O(\log k)$ worst-case time
  per step for both search and deletion, exactly and
  deterministically.  For PAL, this means the KV-cache replacement
  has a bounded per-token latency guarantee with no loss of model
  accuracy.
\end{enumerate}

\begin{figure}[t]
\centering
\begin{tikzpicture}[scale=0.92, every node/.style={font=\scriptsize}]
\begin{scope}[xshift=0cm]
  \draw[->] (0,0) -- (4.4,0) node[below left]{$t$};
  \draw[->] (0,0) -- (0,3.2) node[left]{$u_t$};
  \draw[thick, blue!70!black]
    (0.0,0.6) -- (0.6,2.9) -- (1.2,0.3) -- (1.8,2.3)
    -- (2.4,0.8) -- (3.0,1.9) -- (3.6,1.2);
  \fill[red!80!black]  (0.6,2.9) circle(1.6pt) node[above]{$M_1$};
  \fill[green!50!black](1.2,0.3) circle(1.6pt) node[below]{$m_1$};
  \fill[red!80!black]  (1.8,2.3) circle(1.6pt) node[above]{$M_2$};
  \fill[green!50!black](2.4,0.8) circle(1.6pt) node[below]{$m_2$};
  \fill[red!80!black]  (3.0,1.9) circle(1.6pt) node[above]{$M_3$};
  \node at (2.1,-0.75){(a) input and surviving extrema};
\end{scope}
\begin{scope}[xshift=5.6cm]
  \draw[->] (0,0) -- (3.6,0) node[below left]{$\beta$};
  \draw[->] (0,0) -- (0,3.4) node[left]{$\alpha$};
  \draw[dashed] (0,0) -- (3.1,3.1);
  \node[rotate=45] at (2.35,2.05){$\alpha=\beta$};
  \draw[very thick, purple!80!black]
    (0,2.9) -- (0.32,2.9) -- (0.32,2.3) -- (0.85,2.3)
    -- (0.85,1.9) -- (1.28,1.9) -- (1.28,1.28);
  \fill[purple!80!black] (0.32,2.9) circle(1.6pt)
    node[right=1pt]{$(m_1,M_1)$};
  \fill[purple!80!black] (0.85,2.3) circle(1.6pt)
    node[right=1pt]{$(m_2,M_2)$};
  \fill[purple!80!black] (1.28,1.9) circle(1.6pt)
    node[right=1pt]{$(m_3,M_3)$};
  \node at (1.8,-0.75){(b) staircase interface $=\stack_n$};
\end{scope}
\begin{scope}[xshift=11.2cm]
  \draw[->] (0,0) -- (3.6,0) node[below left]{$\beta$};
  \draw[->] (0,0) -- (0,3.4) node[left]{$\alpha$};
  \draw[dashed] (0,0) -- (3.1,3.1);
  \foreach \x/\y in {0.32/2.9, 0.85/2.3, 1.28/1.9}{
    \draw[red!70, thick] (\x-0.07,\y-0.07) -- (\x+0.07,\y+0.07);
    \draw[red!70, thick] (\x-0.07,\y+0.07) -- (\x+0.07,\y-0.07);}
  \draw[very thick, purple!80!black] (0,3.15) -- (3.15,3.15);
  \fill[red!80!black] (3.15,3.15) circle(1.8pt)
    node[above left=0pt]{$u_{n}=1$};
  \node[align=center] at (1.8,-0.95)
    {(c) global wipe: all $k$ pairs erased\\
     in one step ($\Theta(k)$ output changes)};
\end{scope}
\end{tikzpicture}
\caption{The extremum stack at a glance.
(a)~An input sequence; only the alternating dominant extrema
$M_1 > M_2 > M_3$ (maxima, decreasing) and $m_1 < m_2 < m_3$
(minima, increasing) survive the wiping rule --- their timing and all
non-extremal values are discarded.
(b)~The surviving pairs are the corners of the classical Preisach
staircase interface in the threshold plane; the stack $\stack_n$
\emph{is} this staircase.
(c)~A new global maximum wipes every stored pair at once: the
$\Theta(k)$ output changes are unavoidable (\cref{thm:lower}), but a
finger-tree implementation performs the step in $O(\log k)$ memory
operations (\cref{thm:ftree}).}
\label{fig:overview}
\end{figure}
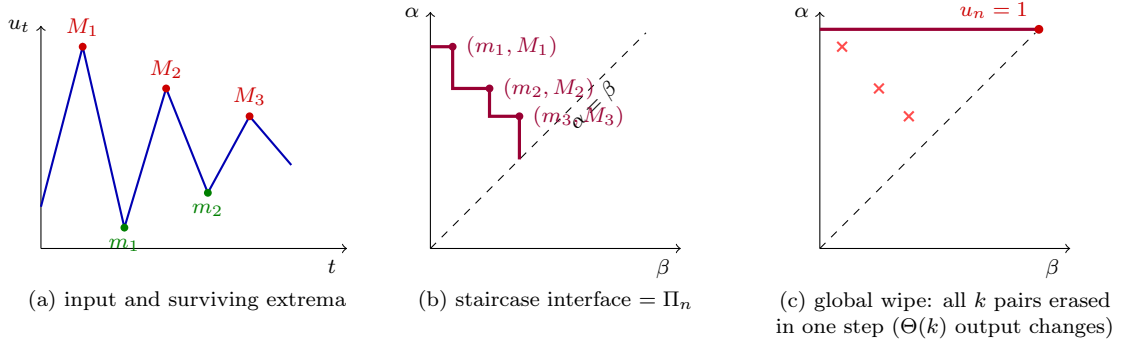

\Cref{fig:overview} summarises the object of study and the source of
the algorithmic difficulty.  The rest of the paper formalises the
three questions in order: \cref{sec:charact} (Q1),
\cref{sec:kolmogorov,sec:shannon} (Q2), \cref{sec:complexity} (Q3).

\subsection{Why the results belong together}

The three question areas are logically interlocked.  The
characterisation theorem (Q1) is the engine of both minimality proofs
(Q2): Kolmogorov sufficiency is a direct corollary
(\cref{cor:upper}), and Shannon minimality uses the characterisation
to reconstruct the stack from any sufficient statistic via a finite
indicator family (\cref{lem:ri,lem:recon}).  The complexity results
(Q3) are stated \emph{relative to} the minimality guarantee: the
objects whose update cost we bound are exactly the
$K_{\RC}$-minimal representations defined by Q2
(\cref{def:repr}), and the lower bound of \cref{thm:lower} would be
vacuous without the minimality notion.  Conversely, the estimation
corollary (\cref{prop:estimation}) combines all three threads:
sufficiency justifies replacing the raw input by the stack process,
minimality guarantees that no smaller replacement exists, and the
online complexity bounds quantify the cost of doing so in a streaming
setting.

\subsection{Contributions}

To the best of our knowledge, the following are new:
\begin{itemize}
\item the characterisation of $\RC$ as exactly the class of stack
  factorisations (\cref{thm:charact}), promoting the classical
  wiping-out property from a feature of Preisach functionals to a
  completeness statement for all computable rate-independent
  functionals;
\item the two-sided Kolmogorov bound
  $K_{\RC}(u_{0:n}) = K(\stack_n) \pm O(1)$ with constants independent
  of $n$ and $k$ (\cref{thm:kolmo}), together with the compression
  lower bound it implies;
\item the Shannon minimality theorem and its uniqueness corollary
  (\cref{thm:minimal,cor:unique});
\item the output-change lower bound for compact exact representations
  (\cref{thm:lower}) and the exact $O(\log k)$ worst-case finger-tree
  stack (\cref{thm:ftree}), the first Preisach stack implementation
  with a non-trivial worst-case guarantee and no approximation error.
\end{itemize}

\subsection{Organisation}

\Cref{sec:related} reviews related work.  \Cref{sec:prelim} fixes
notation and recalls the stack update algorithm.
\Cref{sec:charact} proves the characterisation theorem.
\Cref{sec:kolmogorov,sec:shannon} establish Kolmogorov and Shannon
minimality.  \Cref{sec:complexity} develops the worst-case complexity
picture.  \Cref{sec:estimation} derives the estimation implication.
\Cref{sec:discussion} discusses the relationship between the
characterisations and lists open questions.

\section{Related work}
\label{sec:related}

\paragraph{Sequence modelling and attention mechanisms}
Transformer attention \citep{Vaswani2017} maintains an
$O(n \cdot d_{\mathrm{model}})$ KV-cache; sub-quadratic alternatives
include linear attention \citep{Katharopoulos2020}, state-space
models \citep{Gu2024Mamba}, and sliding-window methods.  PAL
\citep{Frydrych2026PAL} replaces the KV-cache with the extremum
stack, reducing memory to $O(k \cdot d_{\mathrm{model}})$ for
rate-independent tasks.  The present paper provides the theoretical
justification: the stack is the \emph{unique} minimal state for the
rate-independent function class (\cref{thm:charact,thm:kolmo,thm:minimal}),
so no smaller replacement exists, and the finger-tree variant
(\cref{thm:ftree}) bounds worst-case per-token latency at
$O(\log k)$.

\paragraph{Hysteresis and the Preisach operator}
The wiping-out and congruency properties characterise operators
representable in Preisach form \citep{Mayergoyz1991}; the
continuous-time theory of rate-independent evolution is developed in
\citep{Brokate1996,Visintin1994,Krejci1996}.  The reduction of a
Preisach output to the surviving extremum sequence (the ``memory
curve'' or staircase interface) is classical; what has been missing is
a converse: a proof that \emph{nothing but} the stack matters for
\emph{any} computable rate-independent functional, together with
quantitative minimality guarantees.
\Cref{thm:charact,thm:kolmo,thm:minimal} supply this converse.

\paragraph{Algorithmic statistics and Kolmogorov sufficiency}
The theory of algorithmic sufficient statistics
\citep{GacsTrompVitanyi2001,VereshchaginVitanyi2004,LiVitanyi2008}
studies two-part descriptions of individual strings: a model part and
a data-to-model part.  Our setting differs in that minimality is
defined relative to a fixed \emph{query class} $\RC$ rather than a
model class of finite sets or distributions: $K_{\RC}(u_{0:n})$ is the
length of the shortest program answering every $\RC$-query.  This
query-relative notion is closer in spirit to instance optimality in
data structures \citep{Afshani2017} and to the ``useful information''
programme of \citep{VereshchaginVitanyi2004}, but the two-sided
$\pm O(1)$ bound of \cref{thm:kolmo}, with constants uniform in $n$
and $k$, does not follow from those frameworks.

\paragraph{Minimal sufficient statistics and the information
bottleneck}
Shannon-theoretic minimal sufficiency is classical
\citep{CoverThomas2006}; the information bottleneck method
\citep{TishbyPereiraBialek1999} relaxes exact sufficiency to a
Lagrangian trade-off.  \Cref{thm:minimal} establishes \emph{exact}
minimal sufficiency for a nonparametric functional class rather than
for a parametric family, and the uniqueness statement
(\cref{cor:unique}) shows the minimiser is essentially unique.  The
proof route --- reconstructing the statistic from any sufficient
representation via a finite family of indicator queries followed by
the data-processing inequality --- may be of independent interest.

\paragraph{Monotone stacks and Davenport--Schinzel sequences}
The extremum stack is related to, but distinct from, the monotone
stacks used in nearest-greater-element problems and to
Davenport--Schinzel sequences \citep{SharirAgarwal1995,Pettie2015}:
the alternation structure enforced by the wiping rule is precisely an
order-2 alternation constraint on the surviving extrema.  Prior work
on such structures concerns combinatorial length bounds; the
sufficiency, minimality, and worst-case maintenance results proved
here are of a different type and, to our knowledge, new.

\paragraph{De-amortisation and worst-case data structures}
De-amortisation techniques transform amortised bounds into worst-case
bounds for queues \citep{Hood1981}, search trees \citep{Sleator1985},
and priority queues \citep{Brodal1996}.  \Cref{thm:lower} shows that
de-amortising the \emph{deletion} cost of a compact Preisach stack is
impossible in the output-change model, while \cref{thm:ftree} shows
that the \emph{time} cost can nevertheless be de-amortised exactly
using the $O(\log k)$ \textsc{Split} of 2--3 finger trees
\citep{Hinze2006}.  The separation between output changes
(information-theoretic, unavoidable) and memory operations
(implementation-dependent, reducible) appears not to have been made
explicit before.

\paragraph{Time-series and streaming compression}
Dimensionality-reduction and symbolic methods such as PAA
\citep{Keogh2001} and SAX \citep{Lin2003}, streaming compressors such
as Gorilla \citep{Pelkonen2015} and Chimp \citep{Liakos2022}, and
dictionary methods in the Lempel--Ziv family \citep{ZivLempel1977}
offer empirical compression without semantic guarantees; surveys note
per-update latency as a critical constraint in IoT settings
\citep{Chiarot2023}.  The stack-based representation studied here is
a \emph{semantic} compressor: it is lossless exactly for the query
class $\RC$, provably minimal for it (\cref{thm:kolmo,thm:minimal}),
and maintainable with worst-case $O(\log k)$ latency
(\cref{thm:ftree}).  A detailed comparison is given in
\cref{tab:compare}.

\section{Preliminaries}
\label{sec:prelim}

\subsection{Two settings: value domains}
\label{sec:settings}

The paper uses two value domains, matching the two kinds of results.

\emph{Information-theoretic setting
(\cref{sec:charact,sec:kolmogorov,sec:shannon}).}
We fix a resolution $L \geq 2$ and work with sequences over the
discrete grid $\GL = \{0, \Delta, 2\Delta, \ldots, 1\}$,
$\Delta = 1/L$.  Restricting to $\GL$ is necessary here: over $\Rr$,
the indicator family constructed in \cref{sec:kolmogorov} is
uncountable and the Kolmogorov framework does not apply directly.

\emph{Algorithmic setting (\cref{sec:complexity}).}
The complexity results depend only on the \emph{monotone ordering}
structure of the stack (decreasing maxima, increasing minima), which
holds verbatim over any totally ordered value domain.  We therefore
state them over a general totally ordered domain $\mathcal{V}$ (e.g.\
$\Rr$, $\mathbb{Z}$, or $\GL$ with $L$ treated as a growing
parameter), over which the stack depth $k$ can grow unboundedly with
$n$.  This distinction matters: on a \emph{fixed} grid $\GL$, the
strict monotonicity of the maxima forces
$k \leq \lceil (L+1)/2 \rceil$, so $k$ is bounded by a constant
independent of $n$, and worst-case bounds expressed in $k$ would be
trivially $O(1)$ in $n$.  The algorithmic guarantees of
\cref{sec:complexity} are therefore most informative when the value
domain is dense or the quantisation is fine ($L = \Omega(n)$).

For $u_{0:n}$ write $\stack_t = \stack_t(u_{0:t})$ for the extremum
stack at time $t \leq n$.

\subsection{The Preisach operator}

\begin{definition}[Preisach operator]
\label{def:preisach}
For $\alpha \geq \beta$ in $\GL$, the relay
$\relay{\alpha}{\beta}[u](t) \in \{0,1\}$ switches from~$0$ to~$1$
when $u$ exceeds $\alpha$, and from~$1$ to~$0$ when $u$ falls below
$\beta$.  For $\mu \in L^2(\T)$ on the threshold triangle
$\T = \{(\alpha,\beta) : \alpha \geq \beta\}$:
\[
  \mathcal{P}_\mu[u](t) =
  \iint_{\T} \relay{\alpha}{\beta}[u](t)\,\mu(\alpha,\beta)
  \,d\alpha\,d\beta .
\]
\end{definition}

\subsection{The extremum stack}

\begin{definition}[Extremum stack]
\label{def:stack}
Let $\bar{\mathcal{V}} = \mathcal{V} \cup \{\bot, \top\}$ denote the
value domain extended with two boundary symbols, ordered so that
$\bot < v < \top$ for all $v \in \mathcal{V}$.  The \emph{extremum
stack} of $u_{0:n} \in \mathcal{V}^{n+1}$ is the sequence
\[
  \stack_n = [(M_1, m_1), (M_2, m_2), \ldots, (M_k, m_k)]
  \in (\bar{\mathcal{V}} \times \bar{\mathcal{V}})^* ,
\]
where pairs are ordered from \emph{oldest} (bottom, index~1) to
\emph{newest} (top, index~$k$), and pair $(M_i, m_i)$ records a
corner of the classical Preisach staircase interface
\citep{Mayergoyz1991}: $M_i$ is a surviving dominant local maximum
and $m_i$ is the deepest surviving local minimum recorded between
$M_i$ and the next surviving maximum $M_{i+1}$ (for $i < k$), or
since $M_k$ (for $i = k$; see \cref{rem:halfpair}).  The maxima are
strictly decreasing from bottom to top: $M_1 > M_2 > \cdots > M_k$.
The minima are strictly \emph{increasing} from bottom to top:
$m_1 < m_2 < \cdots < m_k$.  Each pair satisfies $M_i > m_i$.  The
boundary symbols appear only in the bottom pair: $M_1 = \top$
encodes ``no dominant maximum recorded yet'' (the input has been
monotone decreasing since $t=0$) and $m_1 = \bot$ encodes ``no
dominant minimum recorded yet''; they are never wiped and act as the
outermost staircase corner.  The depth $k = |\stack_n|$ satisfies
\[
  k \;\leq\; \Bigl\lfloor \tfrac{n}{2} \Bigr\rfloor + 1
  \qquad\text{and, on the grid } \GL,\qquad
  k \;\leq\; \Bigl\lceil \tfrac{L+1}{2} \Bigr\rceil ,
\]
the first because each push requires a direction reversal, the
second because the maxima are strictly decreasing in $\GL$.
\end{definition}

\begin{definition}[Stack process]
\label{def:stack_process}
The \emph{stack process} of $u_{0:n}$ is the sequence
$(\stack_t)_{t=0}^{n}$ of stacks at all intermediate times.  For
$t < n$, $\stack_n$ does not in general determine $\stack_t$, so the
stack process carries strictly more information than the final stack
(cf.\ \cref{rem:final_vs_process}).
\end{definition}

\begin{remark}[Half-pair semantics of the top pair]
\label{rem:halfpair}
The top pair $(M_k, m_k)$ is \emph{provisional} in its second
coordinate: immediately after a new maximum $M_k$ is pushed
(\cref{alg:stack}, line~14), $m_k$ holds the deepest minimum
inherited from the pairs wiped by $M_k$ (or the boundary symbol
$\bot$ if the stack was empty) --- this is exactly the staircase
corner below $M_k$ at that moment.  If the signal subsequently
descends below $m_k$ and reverses, the pop--push mechanism of
\cref{alg:stack} (lines~16--21) replaces the provisional value with
the newly confirmed deeper minimum.  The invariant maintained is that
$\stack_t$, together with the current input $u_t$, always represents
the exact Preisach relay interface; no separate ``half-pair'' object
is needed outside the stack.
\end{remark}

\begin{remark}[Ordering convention]
The decreasing-maxima, increasing-minima ordering follows the
classical Preisach staircase geometry \citep{Mayergoyz1991}: the
outermost loop has the largest maximum and smallest minimum; inner
loops are strictly contained.  In our notation, $(M_1,m_1)$ is the
outermost (oldest) pair and $(M_k,m_k)$ is the innermost (newest).
\end{remark}

The stack is updated in amortised $O(1)$ time per step
(\cref{alg:stack}): each element is pushed and popped at most once.

\begin{algorithm}[t]
\caption{Extremum Stack Update (wiping-out rule)}
\label{alg:stack}
\small
\begin{algorithmic}[1]
\Require Stack $\stack$, previous extremum $e_{\mathrm{prev}}$,
         current direction $d \in \{+1,-1,0\}$,
         new observation $u$
\Ensure Updated $(\stack, e_{\mathrm{prev}}, d)$
\State $d_{\mathrm{new}} \leftarrow \mathrm{sign}(u - e_{\mathrm{prev}})$
\If{$d_{\mathrm{new}} = 0$}
  \State \Return \Comment{No change; flat segment}
\EndIf
\If{$d \neq 0$ \textbf{and} $d_{\mathrm{new}} \neq d$}
  \Comment{Direction reversal: $e_{\mathrm{prev}}$ is a confirmed extremum}
  \If{$d = +1$} \Comment{Previous was a local maximum}
    \State $m_{\mathrm{last}} \leftarrow
      \begin{cases} m_k & \text{if } \stack \neq \emptyset \\
        \bot & \text{otherwise} \end{cases}$
    \While{$\stack \neq \emptyset$ \textbf{and} $M_k < e_{\mathrm{prev}}$}
      \State $m_{\mathrm{last}} \leftarrow m_k$;\quad $\mathrm{pop}(\stack)$
    \EndWhile
    \State $\mathrm{push}(\stack,\,(e_{\mathrm{prev}},\,m_{\mathrm{last}}))$
  \Else \Comment{Previous was a local minimum}
    \State $M_{\mathrm{last}} \leftarrow
      \begin{cases} M_k & \text{if } \stack \neq \emptyset \\
        \top & \text{otherwise} \end{cases}$
    \While{$\stack \neq \emptyset$ \textbf{and} $m_k > e_{\mathrm{prev}}$}
      \State $M_{\mathrm{last}} \leftarrow M_k$;\quad $\mathrm{pop}(\stack)$
    \EndWhile
    \State $\mathrm{push}(\stack,\,(M_{\mathrm{last}},\,e_{\mathrm{prev}}))$
  \EndIf
\EndIf
\State $e_{\mathrm{prev}} \leftarrow u$;\quad $d \leftarrow d_{\mathrm{new}}$
\State \Return $(\stack, e_{\mathrm{prev}}, d)$
\end{algorithmic}
\end{algorithm}

\begin{remark}[Correctness of \cref{alg:stack}]
A push occurs only when a direction reversal is detected, i.e.\ when
$e_{\mathrm{prev}}$ is a \emph{confirmed} local extremum.  During a
monotone run ($d_{\mathrm{new}} = d$), $e_{\mathrm{prev}}$ is updated
but nothing is pushed, preserving the invariants of \cref{def:stack}.
The boundary symbols $\bot, \top$ are pushed only when the stack is
empty and therefore occur only in the bottom pair, as required by
\cref{def:stack}.  When a wipe occurs, $m_{\mathrm{last}}$ takes the
value of the \emph{outermost wiped} pair's minimum --- the deepest
minimum in the wiped window, which is exactly the staircase corner
below the new maximum (cf.\ \cref{rem:halfpair}); when no wipe
occurs, it takes the top surviving pair's provisional minimum, which
the pop--push mechanism then finalises.  The algorithm is initialised
with $\stack = \emptyset$, $e_{\mathrm{prev}} = u_0$, $d = 0$.
\end{remark}

\subsection{Rate-independent functionals}

\begin{definition}[Rate-independence]
\label{def:ri}
A functional $F:\GL^*\to\Rr$ is \emph{rate-independent} if it is
invariant under precomposition with strictly increasing maps:
$F[u\circ\phi] = F[u]$ for every strictly increasing
$\phi:\{0,\ldots,m\} \to \{0,\ldots,n\}$ such that
$u\circ\phi \in \GL^{m+1}$.  Equivalently, $F$ depends on $u$ only
through its \emph{order of extrema}, not their timing.  Write $\RC$
for the class of all computable rate-independent functionals on
$\GL^*$.
\end{definition}

\noindent
A functional is \emph{causal} if $F[u](n)$ depends only on the prefix
$u_0,\ldots,u_n$; since $F$ takes a prefix as input throughout this
paper, causality is implicit.

\begin{remark}[Standard definition]
\label{rem:standard-def}
\cref{def:ri} is the correct discrete-time analogue of the
continuous-time definition of \citep{Brokate1996}: two sequences are
$\RC$-equivalent if and only if one is a monotone
time-reparametrisation of the other up to extremum-equivalence.
(A bijection $\phi:\{0,\ldots,m\}\to\{0,\ldots,n\}$ with $m=n$ must
be the identity, which is why the discrete formulation admits
strictly increasing maps between prefixes of \emph{different}
lengths.)  We do \emph{not} define $\RC$ by the stack factorisation
--- that factorisation is the content of \cref{thm:charact}.
\end{remark}

\subsection{Kolmogorov complexity}

We use standard Kolmogorov complexity \citep{LiVitanyi2008}.  $K(x)$
denotes the length of the shortest self-delimiting program that
outputs $x$ on a fixed universal Turing machine $\mathcal{U}$.  For a
class of functionals $\RC$, we define:

\begin{definition}[$\RC$-query complexity]
\label{def:kR}
\[
  K_{\RC}(u_{0:n})
  = \min\bigl\{|p| : \text{for all } F \in \RC,\;
    \mathcal{U}(p, F) = F[u](n)\bigr\},
\]
the length of the shortest program that, given any $F \in \RC$ as an
additional input, computes the answer $F[u](n)$.
\end{definition}

\section{The Characterisation Theorem}
\label{sec:charact}

\begin{theorem}[Stack characterisation of $\RC$]
\label{thm:charact}
$F \in \RC$ if and only if there exists a computable function
$f:(\GL\times\GL)^*\to\Rr$ such that
$F[u](n) = f(\stack_n(u_{0:n}))$ for all $u_{0:n}\in\GL^*$.
\end{theorem}

\begin{proof}
\textbf{($\Leftarrow$)} If $F[u](n) = f(\stack_n)$, then two
sequences with the same stack have the same output, so $F$ depends
only on the order of extrema and is rate-independent.

\textbf{($\Rightarrow$)} Let $F \in \RC$.  We argue directly from
\cref{def:ri}.  If $\stack_n(u_{0:n}) = \stack_m(v_{0:m})$, then
$u_{0:n}$ and $v_{0:m}$ produce the same ordered sequence of
alternating extrema; any two such sequences are related by a monotone
time reparametrisation (up to insertion of non-extremal points, which
does not change the extremum sequence).  By rate-independence
(\cref{def:ri}), $F[u](n) = F[v](m)$.  Hence $F$ is constant on the
fibres of $\stack_n(\cdot)$ and factors through it:
$F[u](n) = f(\stack_n(u_{0:n}))$ for a well-defined $f$.  Since $F$
is computable and $\stack_n$ is computable from $u_{0:n}$, the
function $f$ is also computable.
\end{proof}

\begin{remark}
For Preisach functionals $F = \mathcal{P}_\mu$, the wiping-out
property \citep{Mayergoyz1991} provides an explicit formula
$F_\mu(\stack_n)$ --- but the factorisation holds for all
$F \in \RC$, not only Preisach functionals.
\end{remark}

\begin{corollary}[Sufficiency of the stack process]
\label{cor:suff_process}
The \emph{stack process} $(\stack_t)_{t=0}^n$ carries strictly more
information than the final stack $\stack_n$: for $t < n$, $\stack_n$
does not in general determine $\stack_t$.  \cref{thm:charact} applies
to each time $t$ separately: $F[u](t) = f_t(\stack_t)$.  For the full
trajectory of queries $(F[u](0),\ldots,F[u](n))$, the sufficient
object is the stack \emph{process}, not the final stack alone.
\end{corollary}

\begin{remark}[Final stack vs.\ stack process]
\label{rem:final_vs_process}
A distinction worth noting: $\stack_n$ is sufficient for the single
query $F[u](n)$ at the \emph{final} time (\cref{thm:charact}).  It is
not sufficient for the full trajectory $(F[u](0),\ldots,F[u](n))$.
All information-theoretic results below refer to the
\emph{final-time} query $F[u](n)$.
\end{remark}

\section{Kolmogorov Minimality}
\label{sec:kolmogorov}

\subsection{Sufficiency and the upper bound}

\begin{proposition}[Stack sufficiency]
\label{prop:suff}
Every $F \in \RC$ satisfies $F[u](n) = f(\stack_n)$ for some
computable function $f$.
\end{proposition}

\begin{proof}
This is the forward direction of \cref{thm:charact}.  For the
connection to the Preisach integral,
$\mathcal{P}_\mu[u](n) = \iint_{\T}
\relay{\alpha}{\beta}[u](n)\,\mu(\alpha,\beta)\,d\alpha\,d\beta$
depends only on $\stack_n$ by the classical wiping property
\citep{Mayergoyz1991}.
\end{proof}

\begin{corollary}[Upper bound]
\label{cor:upper}
$K_{\RC}(u_{0:n}) \leq K(\stack_n) + O(1)$.
\end{corollary}

\begin{proof}
By \cref{prop:suff}, any $F[u](n)$ is computable from $\stack_n$ by
a fixed program of length $O(1)$ (independent of $n$, $k$, $F$).
Since $K$ is defined via a prefix-free universal machine, the
standard invariance property gives $K(f(x)) \leq K(x) + O(1)$ for any
computable $f$ \citep{LiVitanyi2008}.  Embedding the shortest
prefix-free program for $\stack_n$ and appending the $O(1)$-length
post-processor suffices; the prefix-free encoding is self-delimiting
so no additional overhead is needed.
\end{proof}

\begin{remark}[On the overhead: $O(1)$ vs $O(\log K)$]
\label{rem:correction}
An $O(\log n)$ overhead is sometimes suggested informally.  This is
incorrect: the depth $k$ is implicit in the number of pairs and need
not be encoded separately.  A subtler claim of $O(\log K(\stack_n))$
overhead arises from the chain rule
$K(x,y) \leq K(x)+K(y)+O(\log\min\{K(x),K(y)\})$
\citep{LiVitanyi2008}, which applies to pairs of \emph{independently
chosen} objects.  Here, however, $F[u](n)$ is obtained from
$\stack_n$ by a \emph{fixed} computable function; the invariance
theorem $K(f(x)) \leq K(x) + O(1)$ applies directly, giving the
tighter $O(1)$ bound.
\end{remark}

\subsection{Indicator family}

For each pair $(M, m) \in \GL^2$ with $M \geq m$, define the
\emph{indicator functional}:
\begin{equation}
  F_{(M,m)}[u](n)
  = \mathbf{1}\bigl[(M, m) \in \stack_n\bigr].
  \label{eq:indicator}
\end{equation}

\begin{lemma}[Rate-independence of indicators]
\label{lem:ri}
Each $F_{(M,m)}$ defined in \eqref{eq:indicator} belongs to $\RC$.
\end{lemma}

\begin{proof}
Membership of $(M,m)$ in $\stack_n$ depends on $u_{0:n}$ only through
the ordered sequence of surviving extrema, which is invariant under
monotone time reparametrisation.  Hence $F_{(M,m)}$ satisfies
\cref{def:ri}; computability is immediate since $\stack_n$ is
computable from $u_{0:n}$.
\end{proof}

\begin{remark}
$F_{(M,m)}$ is \emph{not} the Preisach relay $\relay{M}{m}[u](n)$,
which indicates whether the relay with thresholds $(M,m)$ is active,
not whether $(M,m)$ appears as a stack entry.  The relay may be
active even when $(M,m) \notin \stack_n$ (if $M$ is not a local
maximum of $u_{0:n}$).  The indicator \eqref{eq:indicator} is a
distinct, valid rate-independent functional.
\end{remark}

\subsection{Minimality theorem}

\begin{lemma}[Stack reconstruction]
\label{lem:recon}
Let $\mathcal{S}$ be any set and $R: \GL^{n+1} \to \mathcal{S}$ a
representation such that every $F \in \RC$ is computable from
$R(u_{0:n})$.  Then there exists a computable function
$\Phi: \mathcal{S} \to (\GL \times \GL)^*$ such that
$\Phi(R(u_{0:n})) = \stack_n$.
\end{lemma}

\begin{proof}
The indicator family $\mathcal{F} = \{F_{(M,m)}\}_{M \geq m,\,
(M,m) \in \GL^2}$ is \emph{finite} (at most $\binom{L}{2} + L =
O(L^2)$ elements) and lies in $\RC$ by \cref{lem:ri}.  Since $R$
suffices for all $F \in \RC$, for each $(M,m) \in \GL^2$ the value
$F_{(M,m)}[u](n) = \mathbf{1}[(M,m) \in \stack_n]$ is computable from
$R(u_{0:n})$.  Define
\[
  \Phi(R(u_{0:n}))
  = \bigl\{(M,m) \in \GL^2 : F_{(M,m)}[u](n) = 1\bigr\}.
\]
This set equals $\stack_n$ by construction (the stack lists its pairs
in the unique order of decreasing maxima, so the member set
determines the sequence), and $\Phi$ is computable (finite
enumeration over $O(L^2)$ pairs).
\end{proof}

\begin{theorem}[Kolmogorov minimality of the extremum stack]
\label{thm:kolmo}
Let $u_{0:n} \in \GL^{n+1}$ and $\stack_n$ its extremum stack.  Then:
\begin{equation}
  K(\stack_n) - O(1)
  \;\leq\; K_{\RC}(u_{0:n})
  \;\leq\; K(\stack_n) + O(1).
  \label{eq:main}
\end{equation}
The extremum stack is, up to an additive constant, the shortest
program that answers every $\RC$-query.
\end{theorem}

\begin{proof}
\textbf{Lower bound.}
By \cref{lem:recon}, there exists a computable $\Phi$ with
$\Phi(R^*(u_{0:n})) = \stack_n$, where $R^*$ is the
$K_{\RC}$-optimal representation.  Therefore:
\[
  K(\stack_n) \leq K(R^*(u_{0:n})) + O(1) = K_{\RC}(u_{0:n}) + O(1),
\]
since $\stack_n$ is computable from the optimal representation by a
program of fixed length $O(1)$ (the description of $\Phi$).
Rearranging gives the lower bound.

\textbf{Upper bound.}
The upper bound follows directly from \cref{cor:upper}: since every
$F[u](n)$ is computable from $\stack_n$ by a fixed $O(1)$-length
program, the prefix-free invariance theorem gives
$K_{\RC}(u_{0:n}) \leq K(\stack_n) + O(1)$.
\end{proof}

\subsection{Tightness of the bounds}

The bounds in \eqref{eq:main} are tight up to the additive constant:
there exist sequences for which
$K_{\RC}(u_{0:n}) \leq K(\stack_n) + O(1)$ and
$K_{\RC}(u_{0:n}) \geq K(\stack_n) - O(1)$ simultaneously (e.g.\ when
$\stack_n$ has a short description that immediately answers all
indicator queries).  The $O(1)$ gap is the unavoidable overhead of
prefix-free encoding on a universal machine \citep{LiVitanyi2008}.

\subsection{Independence from $n$; dependence on $L$}

The $O(1)$ overhead in \eqref{eq:main} is independent of both $n$ and
$k$.  For slowly varying signals (small $k$), the compression ratio
$n/k$ can be arbitrarily large while the optimality gap remains
bounded by a constant.  The stack-based representation is therefore
well-suited for streams with long monotone runs: industrial sensor
data (SCADA), EEG, and financial tick data.

\begin{remark}[Dependence of the constants on the resolution $L$]
\label{rem:L-dep}
The additive constants in \eqref{eq:main} are \emph{not} absolute:
they depend on the grid resolution $L$ (equivalently, on the alphabet
size $|\GL| = L+1$), because the reconstruction program $\Phi$ of
\cref{lem:recon} enumerates the indicator family of size $O(L^2)$,
and the description of $\Phi$ must specify $L$.  A crude bound on the
lower-bound constant is $O(\log L)$ (the length of a self-delimiting
encoding of $L$ plus a fixed enumerator); the upper-bound constant is
likewise $O(\log L)$.  What the theorem asserts is \emph{uniformity
in $n$ and $k$ for each fixed $L$}: the optimality gap does not grow
with the length of the stream or the depth of the stack.
\end{remark}

\subsection{Comparison with existing compression schemes}

\Cref{tab:compare} compares PSTACK-COMPRESS (the stack-based
algorithm implied by \cref{thm:kolmo}) with standard time-series
compression methods.  PSTACK-COMPRESS differs from the rest in
providing a \emph{formal optimality guarantee}: no other listed
method bounds compression length via Kolmogorov complexity for any
function class.  The comparison is complementary to the streaming
compressors discussed in \cref{sec:related}
\citep{Pelkonen2015,Liakos2022,ZivLempel1977}, which target bit-level
redundancy rather than a semantic query class.

\begin{table}[t]
\centering
\caption{Comparison of time-series compression algorithms.
$n$ = sequence length; $k$ = stack depth ($k \leq n$);
$w, s$ = window/segment parameters.
\textdagger{} = asymptotically optimal for class $\RC$.}
\label{tab:compare}
\small
\begin{tabular}{lcccp{3.2cm}}
\toprule
Algorithm & Time & Space & Optimality & Class preserved \\
\midrule
PAA \citep{Keogh2001} & $O(n)$ & $O(w)$ & None &
  $L^2$-approximation \\
SAX \citep{Lin2003} & $O(n)$ & $O(w)$ & None &
  Symbolic, not rate-independent \\
PLR (Piecewise Linear Repr.) & $O(n\log n)$ & $O(s)$ & None &
  $L^\infty$ piecewise linear \\
PIP & $O(n^2)$ & $O(s)$ & None &
  Heuristic important points \\
Swinging Door & $O(n)$ & $O(1)$ & None &
  Linear with tolerance \\
\midrule
\textbf{PSTACK} (ours) & $O(n)$ & $O(k)$ &
  \textbf{\textdagger{} Kolmogorov} &
  \textbf{Full class $\RC$ (rate-independent)} \\
\bottomrule
\end{tabular}
\end{table}

\subsection{Connection to the Preisach operator}

The sufficiency part of \cref{thm:kolmo} is a consequence of the
classical \emph{wiping property} of the Preisach hysteresis operator
\citep{Mayergoyz1991}: a new extremum erases all prior extrema of
smaller magnitude, and the output depends only on the surviving
stack.  The minimality part --- which appears to be new --- shows
that this compression is not merely sufficient but \emph{necessary}.
No computable rate-independent functional ``sees'' less than
$\stack_n$.

\section{Shannon Minimality}
\label{sec:shannon}

Let $u_{0:n}$ be a random sequence drawn from a probability measure
$P$ on $\GL^{n+1}$ (the sequence itself is the random object, not a
parameter).  All statements in this section hold for \emph{every}
$P$; they are informative for non-degenerate $P$ with
$\MI(u_{0:n};\stack_n) > 0$, and reduce to trivial equalities
$0 \leq 0$ when $P$ is a point mass (all mutual informations vanish).
$\En$ and $\MI$ denote Shannon entropy and mutual information
\citep{CoverThomas2006}.

\subsection{Mutual information equality}

\begin{theorem}[MI equality]
\label{thm:mi}
For every $F \in \RC$:
\begin{equation}
  \MI(u_{0:n};\, F[u](n)) = \MI(\stack_n;\, F[u](n)).
  \label{eq:mi_eq}
\end{equation}
\end{theorem}

\begin{proof}
By \cref{thm:charact}, $F[u](n) = f(\stack_n)$ for a computable $f$,
and $\stack_n = \psi(u_{0:n})$ for a computable $\psi$.  Hence both
$F[u](n)$ and $\stack_n$ are deterministic functions of $u_{0:n}$,
giving:
\[
  \En(F[u](n) \mid u_{0:n}) = 0
  \quad\text{and}\quad
  \En(F[u](n) \mid \stack_n) = 0.
\]
Therefore:
\begin{align*}
  \MI(u_{0:n};\,F[u](n))
    &= \En(F[u](n)) - \En(F[u](n)\mid u_{0:n})
     = \En(F[u](n)),\\
  \MI(\stack_n;\,F[u](n))
    &= \En(F[u](n)) - \En(F[u](n)\mid \stack_n)
     = \En(F[u](n)).
\end{align*}
Equality~\eqref{eq:mi_eq} follows.
\end{proof}

\begin{remark}[Triviality and its significance]
\label{rem:trivial}
The equality \eqref{eq:mi_eq} is a direct corollary of
\cref{thm:charact} and the definition of mutual information.  Its
significance lies not in the proof difficulty but in \emph{what it
says}: the stack, despite discarding the timing of all extrema and
the values of all non-extremal points, loses zero mutual information
about any rate-independent query at the final time.  The non-trivial
content is the characterisation theorem itself and, below, the
\emph{minimality} result.
\end{remark}

\subsection{Minimality}

\begin{definition}[$\RC$-sufficient statistic]
\label{def:suff_R}
A random variable $S = S(u_{0:n})$ is \emph{$\RC$-sufficient} if
every $F \in \RC$ satisfies $\En(F[u](n)\mid S) = 0$, i.e.\
$F[u](n)$ is a deterministic function of $S$.
\end{definition}

\begin{remark}
\cref{def:suff_R} formalises ``sufficient for the class $\RC$'' as
the requirement that $S$ determines every $\RC$-query at the final
time $n$.  The stack $\stack_n$ is $\RC$-sufficient by
\cref{thm:charact}.
\end{remark}

\begin{theorem}[Shannon minimality of $\stack_n$]
\label{thm:minimal}
Among all $\RC$-sufficient statistics:
\[
  \MI(u_{0:n};\,\stack_n)
  \;\leq\;
  \MI(u_{0:n};\,S)
  \quad\text{for every $\RC$-sufficient }S.
\]
\end{theorem}

\begin{proof}
Let $S$ be $\RC$-sufficient.  The finite indicator family
$\{F_{(M,m)}\}_{(M,m)\in\GL^2}$ defined by
$F_{(M,m)}[u](n) = \mathbf{1}[(M,m)\in\stack_n]$ lies in $\RC$
(\cref{lem:ri}).  Since $S$ is $\RC$-sufficient, each
$F_{(M,m)}[u](n)$ is a function of $S$.  As the finite collection
$\{F_{(M,m)}\}$ jointly determines $\stack_n$ (\cref{lem:recon}),
there exists a computable $\Phi$ such that $\stack_n = \Phi(S)$.
By the data-processing inequality
\citep[Theorem~2.8.1]{CoverThomas2006}:
\[
  \MI(u_{0:n};\,\stack_n)
  = \MI(u_{0:n};\,\Phi(S))
  \leq \MI(u_{0:n};\,S). \qedhere
\]
\end{proof}

\begin{corollary}[Uniqueness up to $\RC$-equivalence]
\label{cor:unique}
If $S$ is $\RC$-sufficient and
$\MI(u_{0:n};S) = \MI(u_{0:n};\stack_n)$, then $S$ and $\stack_n$ are
functions of each other (i.e.\ they generate the same
$\sigma$-algebra up to $P$-null sets).
\end{corollary}

\begin{proof}
By \cref{thm:minimal}, $\stack_n = \Phi(S)$ and
$\MI(u;\stack_n) \leq \MI(u;S)$.  Equality holds iff $\Phi$ is
injective $P$-a.s., which combined with $S$ being $\RC$-sufficient
(hence determined by $u_{0:n}$) implies $S = \Psi(\stack_n)$ for some
$\Psi$ a.s.
\end{proof}

\begin{remark}[Two minimalities, one object]
\label{rem:two-minimalities}
Kolmogorov minimality (\cref{thm:kolmo}) holds \emph{worst-case for
every individual sequence}; Shannon minimality (\cref{thm:minimal})
holds \emph{in expectation under any probability measure} on the
input.  Neither implies the other: the Kolmogorov bound says nothing
about mutual information under a measure concentrated on compressible
sequences, and the Shannon bound says nothing about individual
adversarial inputs.  Both, however, identify the same object
$\stack_n$ as the canonical information-minimal representation for
$\RC$ --- and both proofs are powered by the same indicator family
(\cref{lem:ri,lem:recon}).
\end{remark}

\section{Worst-Case Update Complexity}
\label{sec:complexity}

\cref{alg:stack} runs in amortised $O(1)$ per step, but an
adversarial input stream can force $O(n)$ total pops followed by a
single push, giving worst-case $\Theta(n)$ cost per step.  This
section determines whether the amortised bound can be de-amortised
--- as is common for de-amortised data structures
\citep{Sleator1985,Brodal1996} --- while preserving the
$K_{\RC}$-minimality guarantee of \cref{thm:kolmo}.  The
representations whose cost we bound are exactly the minimal ones of
the previous sections.

Throughout this section we work in the \emph{algorithmic setting} of
\cref{sec:settings}: the value domain $\mathcal{V}$ is a general
totally ordered set (e.g.\ $\Rr$), or equivalently the grid $\GL$
with the resolution treated as a growing parameter
$L = \Omega(n)$, so that the stack depth $k$ can grow unboundedly
with $n$.  All structural properties used below (strict monotonicity
of maxima and minima, the wiping rule, \cref{alg:stack}) hold
verbatim over any totally ordered domain.  On a fixed finite grid,
$k \leq \lceil (L+1)/2 \rceil$ is a constant in $n$
(\cref{def:stack}) and the bounds below, while still valid, are only
informative as functions of $k$ and $L$ rather than of $n$.

\begin{definition}[$\RC$-minimal representation]
\label{def:repr}
An online algorithm $\mathcal{A}$ maintains an \emph{$\RC$-minimal
representation} if, at every step $n$, its state $S_n$ satisfies:
(i) every $F \in \RC$ is computable from $S_n$, and
(ii) $K(S_n) \leq K(\stack_n) + O(1)$.
\end{definition}

By \cref{thm:kolmo}, the extremum stack itself is the canonical
$\RC$-minimal representation.

\begin{definition}[Compact representation]
\label{def:compact}
An $\RC$-minimal representation $S_n$ is \emph{compact} if it contains
no \emph{dead pairs}: pairs $(M_i, m_i)$ that are dominated by a
newer entry and cannot affect any future $\RC$-query.  The canonical
stack $\stack_n$ is compact by construction.  A \emph{non-compact}
representation may retain dead pairs (marked as inactive), reducing
per-step update cost at the expense of a larger state.
\end{definition}

\subsection{Compact representations: $\Theta(k)$ worst-case in the
output-change model}

\begin{definition}[Output-change model]
\label{def:ocm}
The \emph{output-change cost} of one step of an online algorithm
$\mathcal{A}$ is the number of pairs $(M_i,m_i)$ whose logical status
changes from active to inactive (or vice versa) in that step.  This
model is machine-independent: a linked list pays one pointer
operation per changed pair; a RAM array truncates in $O(1)$ memory
operations but still produces $d$ logical status changes.  The
output-change model is the natural cost measure for rate-independent
computations, where the observable effect is the change in the active
relay set.
\end{definition}

\begin{theorem}[Output-change lower bound for compact
representations]
\label{thm:lower}
Let $\mathcal{A}$ be any deterministic online algorithm maintaining a
compact, exact $\RC$-minimal representation over a totally ordered
value domain $\mathcal{V}$ with $|\mathcal{V}| \geq n+1$.  In the
output-change model, for every $n \geq 2$ there exists a sequence
$u_{0:n-1} \in \mathcal{V}^n$ such that $\mathcal{A}$ incurs
$\Theta(k_{n-1})$ output changes at step $n-1$, where
$k_{n-1} = |\stack_{n-2}| = \Theta(n)$ is the stack depth just before
that step.
\end{theorem}

\begin{proof}
\medskip\noindent\textbf{Adversarial construction.}
Work in the unit interval (rescale as needed) with step
$\delta = 1/(2n)$.  For $t = 0,1,\ldots,n-2$ alternate: for $t$ even
set $u_t = 1 - t\delta$; for $t$ odd set $u_t = t\delta$.  All values
lie in $(0,1)$ and the even values decrease strictly while the odd
values increase strictly, so this creates
$k = \lfloor n/2 \rfloor$ alternating dominant extrema and
$|\stack_{n-2}| = k$.  Set $u_{n-1} = 1$ (global maximum, larger than
all $M_i$).  On the grid $\GL$ the same construction with
$\delta = \Delta$ is valid provided $n \leq L + 1$; for fixed $L$
and $n > L+1$ no such construction exists, since
$k \leq \lceil (L+1)/2 \rceil$ (\cref{def:stack}) --- this is the
regime distinction of \cref{sec:settings}.

\medskip\noindent\textbf{Output changes are $\Theta(k)$.}
At step $n-1$ every pair $(M_i,m_i)$ satisfies $M_i < 1$, so all $k$
pairs are wiped --- each changes status from active to inactive:
$\Theta(k)$ output changes.  Any algorithm maintaining exact
$\RC$-minimality must record these $k$ status changes regardless of
data structure.
\end{proof}

\begin{remark}[Randomised algorithms]
\label{rem:randomized}
The lower bound extends to randomised algorithms maintaining an
\emph{exact} compact representation.  Output changes are an
information-theoretic observable: at step $n-1$ of the construction,
all $k$ pairs deterministically leave the active relay set, and any
representation from which every $\RC$-query is exactly computable
must reflect all $k$ status changes --- randomisation can reorder or
defer the bookkeeping, but not reduce the number of logical changes
without violating exactness.  Randomisation could help only for
\emph{approximate} representations, which are outside the scope of
this paper.
\end{remark}

\begin{remark}[RAM model]
\label{rem:ram}
In a word-RAM model, array truncation costs $O(1)$ memory operations
by adjusting a length counter.  This does not reduce the $\Theta(k)$
output changes: $k$ pairs logically leave the active set, and any
downstream $\RC$-query must observe this.  The output-change lower
bound is therefore model-independent.
\end{remark}

\begin{remark}[Tightness]
\label{rem:tight}
The standard array-based stack with $O(1)$ truncation achieves $O(1)$
memory operations (amortised) while incurring $\Theta(k)$ output
changes in the adversarial case --- matching the lower bound exactly
in the output-change model.
\end{remark}

\subsection{Binary search: the wiping property reduces search, not
deletion}

The wiping property provides \emph{structural information} about the
order of extrema.  Since $M_1 > M_2 > \cdots > M_k$, domination by a
new value $u_{\text{new}}$ is a \emph{monotone} predicate on the
ordered stack.  This structure enables binary search.

\begin{theorem}[Binary search reduces worst-case to $O(\log k + d)$]
\label{thm:binsearch}
Let $\mathcal{A}_{\mathrm{bs}}$ be the online algorithm that stores
$\stack_n$ in a random-access array and uses binary search to find
the domination boundary, followed by bulk deletion of dominated
pairs.  Then:
\begin{enumerate}
  \item \textbf{Boundary search:} $O(\log k)$ comparisons per step,
    exploiting the monotone ordering $M_1 > M_2 > \cdots > M_k$.
  \item \textbf{Deletion cost:} $O(d)$ per step, where $d \leq k$ is
    the number of pairs actually deleted.  Total deletion cost over
    $n$ steps: $O(n)$ (each pair deleted at most once).
  \item \textbf{Worst-case per step:} $O(\log k + d)$.  In the
    adversarial construction of \cref{thm:lower},
    $d = k = \Theta(n)$ and the bound is $\Theta(n)$ --- identical to
    the compact list stack.
  \item \textbf{Amortised cost:} $O(\log k)$ per step (the $O(d)$
    deletion is amortised over insertions).
\end{enumerate}
\end{theorem}

\begin{proof}
\textbf{(1) Binary search.}
Since $M_1 > M_2 > \cdots > M_k$, the predicate
$M_i < u_{\text{new}}$ is monotone in $i$.  Binary search on the
random-access array finds $i^* = \min\{i : M_i < u_{\text{new}}\}$ in
$\lceil \log_2 k \rceil$ comparisons.

\textbf{(2) Deletion.}
Pairs $i^*, i^*+1, \ldots, k$ are deleted by setting the stack
pointer to $i^*-1$ (array truncation: $O(1)$ if implemented as a
length variable; $O(d)$ if memory is reclaimed immediately).  Each
pair is inserted once and deleted at most once, giving $O(n)$ total
deletion cost.

\textbf{(3) Worst-case.}
In the adversarial construction, $u_{n-1} = 1 > M_1$, so $i^* = 1$
(found in $O(1)$ by binary search) but $d = k$ deletions follow.
Worst-case cost: $O(\log k) + O(k) = O(k) = \Theta(n)$.  Binary
search reduces the \emph{search} cost but not the \emph{deletion}
cost.

\textbf{(4) Amortised.}
The total cost over $n$ steps is $O(n \log k) + O(n) = O(n \log n)$,
giving $O(\log n)$ amortised --- identical to the standard stack's
$O(1)$ amortised up to the $\log$ factor.  The standard stack
achieves $O(1)$ amortised by charging deletion to insertion; binary
search pays an extra $O(\log k)$ per step for the boundary search.
\end{proof}

\begin{remark}[The structural information of wiping]
\label{rem:wipe_info}
The wiping property does reduce the \emph{search} complexity from
$O(k)$ to $O(\log k)$: monotone ordering enables binary search.
However, it does not reduce the \emph{deletion} complexity: a compact
representation must physically remove all dominated pairs, and their
count $d$ is unaffected by the search strategy.  The $\Theta(k)$
worst-case of \cref{thm:lower} is therefore tight for compact
representations regardless of search strategy.  For non-compact
representations with lazy deletion, the per-step cost drops to
$O(\log k)$ at the cost of $O(k_{\text{dead}})$ extra state ---
violating exact $K_{\RC}$-minimality.
\end{remark}

\subsection{Finger-tree stack with $O(\log k)$ exact worst-case}
\label{sec:ftree}

\cref{thm:binsearch} bounds boundary \emph{search} at $O(\log k)$ but
leaves physical deletion at $\Theta(d)$.  We now show that both
search \emph{and} deletion can be performed in $O(\log k)$ worst-case
time using a \emph{finger tree} \citep{Hinze2006} --- an exact
representation with no approximation.

\begin{definition}[Finger-tree Preisach stack]
\label{def:ftree}
A \emph{finger-tree Preisach stack} stores the pairs
$[(M_1,m_1),\ldots,(M_k,m_k)]$ in a 2--3 finger tree
\citep{Hinze2006} in stack order, with the maxima $M_i$ (strictly
decreasing in $i$) serving as monotone keys for search.\footnote{The
2--3 finger tree of Hinze and Paterson is a \emph{sequence}
structure supporting \textsc{Cons}, \textsc{Snoc}, and
\textsc{Split} on a monoidal measure.  We use the standard
ordered-sequence instantiation: nodes are annotated with the
\emph{minimum} key of their subtree --- equivalently its rightmost
key, since the $M_i$ decrease along the sequence --- under the
monoid $(\min, +\infty)$, which is associative with identity
$+\infty$.  The search predicate ``the running minimum has dropped
below $u_{\mathrm{new}}$'' is monotone in the prefix order
(extending a prefix can only lower its minimum), so \textsc{Search}
is the predicate-guided \textsc{Split} of \citep{Hinze2006}.
Annotating with the \emph{maximum} key would not work: a prefix can
have a large leftmost maximum while containing deeper elements below
the threshold.  For the symmetric minimum-reversal branch, the same
nodes carry a second annotation --- the maximum of the minima $m_i$
under the monoid $(\max, -\infty)$ --- so a single tree with a
product-monoid measure
$\bigl((\min,+\infty)\times(\max,-\infty)\bigr)$ serves both
branches.  Equivalently, the structure may be viewed as a
balanced 2--3 search tree with finger access to both ends.}  The
tree supports the
following operations, each in $O(\log k)$ worst-case time (the
end-access operations in $O(1)$ amortised):
(a) \textsc{Search}$(u)$: find the split point
    $i^* = \min\{i : M_i < u\}$ by monotone (measure-guided) search;
(b) \textsc{Split}$(i^*)$: split the tree at $i^*$, returning the
    pair of trees
    $(T_{\mathrm{pre}}, T_{\mathrm{suf}})$ holding
    $[(M_1,m_1),\ldots,(M_{i^*-1},m_{i^*-1})]$ and
    $[(M_{i^*},m_{i^*}),\ldots,(M_k,m_k)]$ respectively, in
    $O(\log k)$ time (not $O(d)$);
(c) \textsc{Push}$(M,m)$: append one pair at the top end in $O(1)$
    amortised, $O(\log k)$ worst-case;
(d) \textsc{ViewLast}$()$ / \textsc{ViewFirst}$()$: read the
    top-end / bottom-end element without removing it.
\end{definition}

\begin{algorithm}[t]
\caption{Finger-Tree Extremum Stack Update}
\label{alg:ftree}
\small
\begin{algorithmic}[1]
\Require Finger tree $T$, previous extremum $e_{\mathrm{prev}}$,
         current direction $\mathit{dir} \in \{+1,-1,0\}$, new value $u$
\Ensure Updated $(T, e_{\mathrm{prev}}, \mathit{dir})$
\State $\mathit{dir}_{\mathrm{new}} \leftarrow \mathrm{sign}(u - e_{\mathrm{prev}})$
\If{$\mathit{dir}_{\mathrm{new}} = 0$}
  \State \Return \Comment{Flat: no extremum}
\EndIf
\If{$\mathit{dir} \neq 0$ \textbf{and} $\mathit{dir}_{\mathrm{new}} \neq \mathit{dir}$}
  \Comment{Reversal: $e_{\mathrm{prev}}$ is confirmed extremum}
  \State $i^* \leftarrow T.\textsc{Search}(e_{\mathrm{prev}})$
    \Comment{$O(\log k)$: find domination boundary}
  \State $(T_{\mathrm{pre}}, T_{\mathrm{suf}}) \leftarrow T.\textsc{Split}(i^*)$
    \Comment{$O(\log k)$: separate dominated suffix}
  \If{$T_{\mathrm{suf}} \neq \emptyset$}
    \Comment{Wipe occurred}
    \State $m_{\mathrm{last}} \leftarrow T_{\mathrm{suf}}.\textsc{ViewFirst}().m$
      \Comment{outermost wiped pair's minimum $= m_{i^*}$}
    \State discard $T_{\mathrm{suf}}$
      \Comment{released as a unit; $O(\log k)$}
  \ElsIf{$T_{\mathrm{pre}} \neq \emptyset$}
    \Comment{No wipe}
    \State $m_{\mathrm{last}} \leftarrow T_{\mathrm{pre}}.\textsc{ViewLast}().m$
      \Comment{top surviving pair's minimum $= m_{i^*-1}$}
  \Else
    \State $m_{\mathrm{last}} \leftarrow \bot$
      \Comment{empty stack: boundary symbol}
  \EndIf
  \State $T \leftarrow T_{\mathrm{pre}}$;\quad
    $T.\textsc{Push}(e_{\mathrm{prev}}, m_{\mathrm{last}})$
    \Comment{$O(\log k)$: insert new pair}
\EndIf
\State $e_{\mathrm{prev}} \leftarrow u$;\quad $\mathit{dir} \leftarrow \mathit{dir}_{\mathrm{new}}$
\State \Return $(T, e_{\mathrm{prev}}, \mathit{dir})$
\end{algorithmic}
\end{algorithm}

\begin{theorem}[Finger-tree stack: $O(\log k)$ exact worst-case]
\label{thm:ftree}
\cref{alg:ftree} maintains an exact, compact $\RC$-minimal
representation and achieves:
\begin{enumerate}
  \item \textbf{Worst-case time per step:} $O(\log k)$, where
    $k = |\stack_n|$ is the current stack depth.  In particular, a
    global wipe ($d = k$ pairs dominated) costs $O(\log k)$, not
    $O(k)$.
  \item \textbf{Output changes:} $\Theta(d)$ per step (matching the
    lower bound of \cref{thm:lower}).
  \item \textbf{Space:} $O(k)$.
  \item \textbf{Exactness:} The output equals $\stack_n$ exactly; no
    approximation error.
\end{enumerate}
\end{theorem}

\begin{proof}
\textbf{Time.}
The \textsc{Search} step finds $i^*$ in $O(\log k)$ by monotone
(measure-guided) search on the finger tree \citep{Hinze2006}.  The
\textsc{Split} step separates the dominated suffix $[i^*, k]$ in
$O(\log k)$ by splitting the 2--3 tree at position $i^*$ --- no
per-element traversal: the split visits $O(\text{height}) =
O(\log k)$ nodes, and the $d$ nodes of $T_{\mathrm{suf}}$ are
subsequently released as a unit (collected lazily by the memory
manager).  This distinguishes a finger tree from a plain array or
linked list, where suffix removal requires $\Theta(d)$ individual
operations.  \textsc{ViewFirst} on $T_{\mathrm{suf}}$ and
\textsc{ViewLast} on $T_{\mathrm{pre}}$ are end-access operations,
$O(1)$ amortised and $O(\log k)$ worst-case.  \textsc{Push} costs
$O(\log k)$ worst-case.  Total per step: $O(\log k)$.

\textbf{Correctness (equivalence with \cref{alg:stack}).}
The value $m_{\mathrm{last}}$ computed by \cref{alg:ftree} coincides
with that of \cref{alg:stack} in all three cases.  If a wipe occurs
($T_{\mathrm{suf}} \neq \emptyset$), the while-loop of
\cref{alg:stack} pops pairs $k, k-1, \ldots, i^*$ in order, leaving
$m_{\mathrm{last}} = m_{i^*}$ --- the outermost wiped pair's minimum,
which is exactly $T_{\mathrm{suf}}.\textsc{ViewFirst}().m$.  If no
wipe occurs and the stack is non-empty, \cref{alg:stack} initialises
$m_{\mathrm{last}} = m_k$ of the top surviving pair, which is
$T_{\mathrm{pre}}.\textsc{ViewLast}().m$ with $i^* = k+1$.  If the
stack is empty, both algorithms use the boundary symbol $\bot$.  The
minimum-reversal branch (omitted from \cref{alg:ftree} for brevity)
is symmetric, with the roles of $M$ and $m$ exchanged; its
\textsc{Search} runs over the strictly \emph{increasing} minima
$m_i$ and is guided by the second component of the product-monoid
annotation (the subtree maximum of the $m_i$; see the footnote to
\cref{def:ftree}), so no separate tree is needed --- both branches
operate on the same structure in $O(\log k)$.

\textbf{Output changes.}
The $d$ pairs in the discarded suffix change logical status from
active to inactive: $\Theta(d)$ output changes, matching
\cref{thm:lower}.  The \textsc{Split} achieves this in $O(\log k)$
memory operations because the tree represents the suffix implicitly
--- the root of the discarded subtree is released in one pointer
update.

\textbf{Space.} The 2--3 finger tree uses $O(k)$ nodes for $k$ pairs.

\textbf{Exactness.} The output $T$ represents $\stack_n$ exactly
after each step; no $\varepsilon$-threshold is applied.
\end{proof}

\begin{remark}[Comparison with buffer-based approaches]
A lazy-propagation buffer achieves $O(\log n)$ worst-case per
memory-operation step but incurs $\Theta(k)$ output changes during a
flush when a new global maximum dominates all existing pairs, and
retains dead state that violates compactness.  The finger-tree
approach achieves $O(\log k)$ both in memory operations \emph{and} is
exact --- at the cost of a more complex data structure.
\end{remark}

\subsection{Relation to the de-amortisation literature}

The standard technique for de-amortising data structures is to spread
the work of expensive operations across future cheap ones
\citep{Hood1981,Sleator1985}.  \cref{thm:lower} shows that
de-amortisation of the \emph{deletion} cost is fundamentally
impossible for compact Preisach stacks: the physical removal of $d$
dominated pairs requires $\Theta(d)$ output changes regardless of any
search strategy, and cannot be deferred without violating
compactness.

\cref{thm:binsearch} shows, however, that the wiping property
\emph{does} provide structural information that reduces the
\emph{search} cost: monotone ordering of maxima
$M_1 > \cdots > M_k$ enables $O(\log k)$ boundary detection.
\cref{thm:ftree} shows that a finger tree \citep{Hinze2006} achieves
$O(\log k)$ worst-case for \emph{both} search and deletion via the
\textsc{Split} operation, which discards a dominated suffix in
$O(\log k)$ time without traversing individual nodes.  This contrasts
with binary search trees where de-amortisation is possible
\citep{Brodal1996} but typically requires approximate or
probabilistic techniques; the finger tree achieves exact,
deterministic $O(\log k)$ worst-case.

The three-level picture (\cref{tab:summary}) is therefore sharp:
\begin{itemize}
  \item Array with truncation: $O(1)$ memory ops, $\Theta(k)$ output
    changes worst-case --- optimal for throughput.
  \item Binary search on array: $O(\log k + d)$ memory ops, same
    output changes --- reduces search cost only.
  \item Finger tree (new): $O(\log k)$ memory ops, same output
    changes --- optimal worst-case latency, exact, no approximation.
\end{itemize}

\begin{table}[t]
\centering
\caption{Complexity of Preisach stack variants.
$k$ = current depth; $d$ = output changes in one step.
Memory ops = RAM operations; Output changes = logical status changes.}
\label{tab:summary}
\small
\begin{tabular}{lcccc}
\toprule
Representation & Mem.\ ops/step & Output changes/step
  & Amortised & $K_{\RC}$-minimality \\
\midrule
Array (truncation) & $O(1)$ & $\Theta(d)$ & $O(1)$ & Exact \\
Binary search (array) & $O(\log k + d)$ & $\Theta(d)$
  & $O(\log n)$ & Exact \\
\textbf{Finger tree (new)} & $\mathbf{O(\log k)}$ & $\Theta(d)$
  & $O(\log k)$ & \textbf{Exact} \\
\bottomrule
\end{tabular}
\end{table}

\subsection{Practical implications}

\Cref{tab:summary} leads to a direct engineering decision:
\emph{choose the array stack for throughput; choose the finger-tree
stack for latency guarantees}.  Both are exact --- the approximation
trade-off has been eliminated.

\textbf{PAL KV-cache replacement.}
Replacing the KV-cache in a Preisach attention architecture
\citep{Frydrych2026PAL} with the extremum stack reduces memory from
$O(n \cdot d_{\mathrm{model}})$ to $O(k \cdot d_{\mathrm{model}})$.
\Cref{thm:kolmo,thm:minimal} guarantee that no smaller replacement
exists for the rate-independent function class; the finger-tree
variant (\cref{thm:ftree}) adds a worst-case $O(\log k)$ per-token
latency guarantee with no change in model output.  On edge hardware
with per-token latency budgets, this eliminates the $\Theta(k)$
latency spikes that occur when a new global extremum wipes the
entire stack.

\textbf{Signal-profile-dependent choice.}
For slowly varying signals ($k \ll n$, rare global wipes), the array
stack suffices: amortised $O(1)$ dominates and the adversarial case
never occurs in practice.  For signals with frequent global extrema
(flash crashes, mechanical impacts, epileptic EEG bursts), the finger
tree eliminates latency spikes without any loss of accuracy.  The
choice is purely about latency requirements, not accuracy.

\textbf{Streaming compression and real-time systems.}
Beyond sequence modelling, the results apply to any system that
processes rate-independent streams online.  Latency is identified as
a critical constraint in IoT time-series compression
\citep{Chiarot2023}; for devices with strict per-update latency
budgets, the finger tree reduces worst-case update cost from
$\Theta(k)$ to $O(\log k)$ memory operations --- a factor of
$\sim 700$ for $k = 10^4$ extrema --- while preserving the Kolmogorov
optimality of \cref{thm:kolmo} exactly.  For hard real-time systems
(PLC controllers, SCADA), the finger tree bounds processing latency
to $O(\log k)$ regardless of how many relay pairs are logically
wiped.

\section{Estimation Implication}
\label{sec:estimation}

The classical approach to estimating $\mu$ from observations
$Y_{0:n}$ is non-negative least squares (NNLS)
\citep{Mayergoyz1991}:
\begin{equation}
  \hat\mu = \arg\min_{\mu\geq 0}
  \sum_{t=0}^n (Y_t - \mathcal{P}_\mu[u](t))^2.
  \label{eq:nnls}
\end{equation}

\begin{proposition}[Estimation via the stack process]
\label{prop:estimation}
The NNLS system~\eqref{eq:nnls} can be assembled from the
\emph{stack process} $(\stack_t)_{t=0}^n$ without retaining
$u_{0:n}$:
\begin{enumerate}
  \item At each time $t$, $\mathcal{P}_\mu[u](t) = F_\mu(\stack_t)$
    by \cref{thm:charact}, so each residual
    $(Y_t - F_\mu(\stack_t))$ is computable from $\stack_t$ and
    $\mu$.
  \item The stack $\stack_t$ is maintained online in $O(k_t)$ space,
    where $k_t = |\stack_t|$, with per-step cost as in
    \cref{tab:summary}: $O(1)$ amortised (array) or $O(\log k_t)$
    worst-case (finger tree, \cref{thm:ftree}).
  \item The relay state lookup $F_\mu(\stack_t)$ costs $O(L^2)$ per
    step.  The NNLS Gram matrix has size $L^2 \times L^2 = O(L^4)$,
    independent of $n$; this is the dominant memory cost and is
    unchanged by the stack approach.  The saving is in \emph{input
    history storage}: the full input history $u_{0:n}$ requires
    $O(n)$ space, whereas the stack process requires
    $O(\max_t k_t)$ space --- a reduction by a factor of
    $n/\max_t k_t$, which can be $\Omega(n)$ for slowly varying
    signals.
\end{enumerate}
\end{proposition}

\begin{remark}[Final stack vs.\ stack process]
The space reduction comes from \emph{online stack maintenance}
(replacing the raw input history with a stack of alternating extrema
at each step), \emph{not} from compressing the trajectory to the
final stack $\stack_n$ alone.  The number of NNLS equations remains
$n$; what is reduced is the memory required to evaluate each
equation (cf.\ \cref{rem:final_vs_process}).
\end{remark}

\section{Discussion and Conclusion}
\label{sec:discussion}

\subsection{Three characterisations of one object}

The results form a closed picture.  \Cref{thm:charact} says the stack
is \emph{qualitatively} right: rate-independence \emph{is} stack
factorisation.  \Cref{thm:kolmo,thm:minimal} say it is
\emph{quantitatively} right, in the two standard senses of
information minimality --- per-instance and in expectation --- and
\cref{cor:unique} says it is essentially the \emph{only} such object.
\Cref{thm:lower,thm:ftree} say the minimal object is also
\emph{cheap}: the unavoidable per-step cost is $\Theta(d)$ output
changes, and everything else can be brought down to $O(\log k)$
worst-case, exactly.

The two minimality results are complementary in a precise sense
(\cref{rem:two-minimalities}): Kolmogorov minimality holds worst-case
for every individual sequence; Shannon minimality holds in
expectation under any probability measure on the input.  Neither
implies the other, but both identify $\stack_n$ as the canonical
information-minimal representation for $\RC$, and both proofs are
powered by the same finite indicator family.  The complexity analysis
then quantifies the cost of maintaining this canonical object online:
compact exact representations incur $\Theta(k)$ output changes per
step in the adversarial case --- unavoidable in any model; binary
search reduces boundary detection to $O(\log k)$ comparisons but
leaves deletion at $\Theta(d)$ per step; a finger-tree implementation
achieves $O(\log k)$ worst-case for both search and deletion via the
\textsc{Split} operation, while maintaining exact $\RC$-minimality ---
the first Preisach stack implementation with a non-trivial worst-case
guarantee and no approximation error.

\subsection{Open questions}

\begin{enumerate}
  \item \emph{Approximate representations.}  The output-change lower
    bound of \cref{thm:lower} extends to randomised algorithms for
    exact representations (\cref{rem:randomized}).  Can
    \emph{approximate} compact representations --- e.g.\ retaining
    only pairs with loop amplitude $M_i - m_i \geq \varepsilon$ ---
    achieve $o(k)$ output changes per step, and what is the resulting
    error on the Preisach output as a function of $\varepsilon$ and
    $\mu$?
  \item \emph{Space constants.}  What is the optimal \emph{space}
    overhead of the finger-tree stack relative to the array stack?
    The 2--3 finger tree uses $O(k)$ nodes with a constant factor
    larger than a plain array; can a more cache-friendly layout
    reduce this constant?
  \item \emph{Vector extension.}  Does an analogue of
    \cref{thm:kolmo,thm:minimal,thm:lower,thm:ftree} hold for the
    vector Preisach operator \citep{Brokate1996,Frydrych2019} with
    two-dimensional input signals?  The joint representation
    $(\stack^x_n, \stack^y_n)$ of two independent stacks is a
    computable pairing, so the overhead for answering all vector
    $\RC$-queries should remain $O(1)$ by the same invariance
    argument; the joint finger tree would support
    $O(\log(k_x + k_y))$ worst-case \textsc{Split} on the product
    space.  Does this extend to the full vector Preisach measure
    $\mu(\alpha,\beta,\theta)$ over directions $\theta \in [0,2\pi)$?
  \item \emph{Trajectory sufficiency.}  Can the stack-process
    sufficiency (\cref{cor:suff_process}) be given a formal
    Fisher--Neyman characterisation for $\mu$-estimation with full
    trajectory data?
\end{enumerate}

\section*{Declaration on the use of generative AI}
During the preparation of this work the author used AI-assisted
tools for grammar checking and \LaTeX\ formatting.
All mathematical content, proofs, and claims are the sole
responsibility of the author.

\section*{Declaration of competing interests}
The author declares no competing financial or non-financial
interests.

\section*{Acknowledgements}
The author thanks the editors and reviewers of earlier, shorter
versions of this material for comments that led to the present
unified presentation, and in particular for remarks that sharpened
the statement of the indicator-family lemma (\cref{lem:ri}).

\bibliographystyle{elsarticle-num}
\bibliography{references_unified}

\end{document}